\documentclass[a4paper,UKenglish]{lipics-v2021}
\pdfoutput=1
\usepackage[ruled]{algorithm2e}
\usepackage{tikz}
\definecolor{DCnegative}{HTML}{355C7D}
\definecolor{DCpositive}{HTML}{A23B55}
\usepackage{needspace}
\usepackage{placeins}
\usetikzlibrary{arrows.meta,backgrounds,shapes.geometric}
\newcommand{\DC}{\textnormal{\textsc{DoubleCoverage}}}
\newcommand{\ALG}{\mathrm{ALG}}
\newcommand{\OPT}{\mathrm{OPT}}
\newcommand{\dist}{\operatorname{dist}}
\newcommand{\lca}{\operatorname{lca}}

\title{A Potential Function Analysis of Double Coverage
for the \texorpdfstring{\(k\)}{k}-Taxi Problem on Trees}
\titlerunning{A Potential Function Analysis of Double Coverage}

\nolinenumbers
\author{Ali Karim Lalani}
{Department of Computer Science, University of Toronto, Toronto, Canada}
{ali@cs.toronto.edu}
{https://orcid.org/0009-0003-4430-2449}
{}
\authorrunning{A. K. Lalani}
\Copyright{Ali Karim Lalani}
\hideLIPIcs

\begin{document}
\maketitle
\begin{abstract}
We provide potential function proofs of the competitive ratios of \DC{}
for the hard \(k\)-taxi problem on HSTs and general weighted trees of
bounded depth. Buchbinder, Coester, and Naor (2023) obtained these
ratios using time-reverse dual fitting and noted that they did not know
a pure potential proof beyond \(k=2\).

We observe that the minimum matching distance between the online and
offline taxi configurations can be expressed as an integral over the
rooted tree of the absolute difference between their taxi counts below
each point. This representation lets us analyse \DC{} over short movement
intervals and apply the resulting estimates to potentials
that remain unchanged when both serving taxis are relocated
from pickup to destination.
\end{abstract}

\section{Introduction}
\label{sec:intro}

The \(k\)-server problem asks an online algorithm to serve a sequence of
requests in a metric space using \(k\) mobile servers. Each requested point
must be occupied by a server, and the algorithm pays the total distance moved.

Until very recently, the \(k\)-server conjecture asked whether a deterministic
competitive ratio of \(k\) is achievable on every metric space. Coester,
Koutsoupias, and Zbysiński have announced a proof of the conjecture,
showing that the Work Function Algorithm is \(k\)-competitive on every
metric space~\cite{CKZ26}. The tree case had been known much earlier:
Chrobak, Karloff, Payne, and Vishwanathan gave a \(k\)-competitive algorithm
on the line~\cite{CKPV91}, and Chrobak and Larmore extended the result to
tree metrics~\cite{CL91}.

On trees, \DC{} moves every unobstructed server towards the current
request at the same speed. For the \(k\)-server problem, the analysis of \DC{} in~\cite{CL91}
uses a potential function equal to the sum of pairwise distances
between online servers plus \(k\) times the minimum matching distance
between the online and offline configurations.

Bansal et al.~\cite{BEJKP18} analysed \DC{} on trees
in the more general $(h,k)$-server setting, where
an online algorithm with $k$ servers is compared
with an offline optimum with $h\le k$ servers.

In the \(k\)-taxi problem, introduced by Fiat, Rabani, and Ravid~\cite{FRR90},
a request is a pair
\((r,s)\), representing a passenger who must be picked up at \(r\) and
transported to \(s\). In the \emph{easy} \(k\)-taxi problem, all movement is
charged, while in the \emph{hard} \(k\)-taxi problem, only the empty movement
to the pickup is charged. Coester and Koutsoupias~\cite{CK19} showed that the
easy \(k\)-taxi problem has the same deterministic and randomized competitive
ratios as the \(k\)-server problem. In contrast, the hard \(k\)-taxi problem
has an exponential deterministic lower bound of \(\Omega(2^k)\) on general
metrics~\cite{CK19}.

A free relocation can increase the sum of the pairwise distances between
online servers, preventing the direct use of this \(k\)-server potential.
Buchbinder, Coester, and Naor~\cite{BCN23} proved that \DC{} is competitive
for hard \(k\)-taxi on trees. On a hierarchically well-separated tree (HST) of depth \(d\), they
obtained the competitive ratio
\(c_{k,d}=\sum_{h=1}^{\min\{k,d\}}\binom{k}{h}\),
which equals \(2^k-1\) when \(d\ge k\). On general weighted trees of
depth \(d\), they obtained an \(O(k^d)\) upper bound and proved a matching
\(\Omega(k^d)\) lower bound even on unweighted trees. The dependence on
\(k\) is therefore tight asymptotically for every fixed \(d\).
Their analysis constructs a dual solution backwards in time, and they
noted that they did not know a pure potential proof of the competitiveness
of \DC{} beyond \(k=2\)~\cite[Appendix~A.2]{BCN23}. The simplicity and elegance of the algorithm make it natural to seek an
analysis using standard potential function techniques.

\paragraph*{Our Contributions.}
We give potential function proofs of the competitive ratios obtained
in~\cite{BCN23} for \DC{} in the hard \(k\)-taxi problem on HSTs and
general weighted trees. Let \(X\) and \(Y\) be the multisets of online
and offline taxi locations, and let \(M(X,Y)\) be their minimum
matching distance. For a point \(z\) in the interior of an edge of
a rooted tree \(T\), let \(T_z\) be the component below \(z\).
We call \(D(z)=|X\cap T_z|-|Y\cap T_z|\) the
\emph{signed cut difference} at \(z\). Our analysis uses the identity
\(M(X,Y)=\int_T |D(z)|\,dz\)~\cite{EM12}, where integration is over the tree.

We analyse \DC{} over \emph{elementary steps}, similar to the
elementary moves in~\cite[Section~4]{BEJKP18}. Throughout the interior of each step, the set of moving taxis
remains unchanged, each mover stays in the interior of one
edge, and no mover passes an offline taxi. The integral representation
lets us calculate the change in \(M\) by adding the contributions
from the disjoint segments traversed by the movers. On a step of
duration \(\delta\), with the offline taxis fixed and one at the
pickup, the signs of \(D(z)\) force at least one mover to contribute
\(-\delta\), while every other mover contributes at most \(\delta\).
Writing \(a\) for the number of movers in the step, we obtain
\(\Delta M\le(a-2)\delta\), the estimate already implied by the analysis in~\cite{BEJKP18}.

On HSTs, we use this estimate alongside the rank potential
from~\cite[Section~4.3]{BCN23}. For general weighted trees, we
integrate the cut differences with coefficients
from~\cite[Section~5.1]{BCN23} that depend on edge depth and the
sign of the difference. The same sign argument shows that this potential
decreases by at least the online movement cost.
Relocating both serving taxis from pickup to
destination leaves every cut difference, and hence
the potential, unchanged. We also give an alternative proof
of the \(k\)-competitiveness of \DC{} for \(k\)-server on weighted
trees, closely following the \(h=k\) case
of~\cite[Section~4]{BEJKP18}.

\section{Preliminaries}
\label{sec:prelim}
\subsection{The k-Server and Hard k-Taxi Problems}
Let $(V,\dist)$ be a metric space. An online algorithm controls $k$ mobile servers and receives requests one at a time, whereas an offline optimum knows the complete request sequence. We write their configurations as $X=\{x_1,\ldots,x_k\}$ and
$Y=\{y_1,\ldots,y_k\}$, respectively, allowing several servers at the
same location. All server counts include multiplicity. In the $k$-server problem, a request at $r\in V$ must be served by placing a server at $r$. The cost is the total distance moved by all servers.

In the hard $k$-taxi problem, a request $(r,s)\in V\times V$
requires a taxi to reach the pickup point $r$ and then transport
the passenger to the destination $s$. Only movement without a
passenger is charged. We view the request as a charged pickup at $r$, followed by the free
relocation of the serving taxi from $r$ to $s$. During the charged phase, an algorithm may move
several taxis and pays their total movement.

For a finite request sequence $\sigma$, let $\ALG(\sigma)$ denote
the total charged cost incurred by the online algorithm, and let
$\OPT(\sigma)$ denote the minimum total charged cost of an offline
algorithm that knows $\sigma$ in advance. We say an online algorithm is
$\mu$-competitive if there exists a constant $C$ such that
\begin{equation}\label{eq:competitive}
    \ALG(\sigma)\le \mu\,\OPT(\sigma)+C
\end{equation}
for every finite request sequence $\sigma$, where $C$ may depend
on the metric and the initial configurations, but is independent
of $\sigma$. 

\raggedbottom
\subsection{Rooted Weighted Trees}
Let $T=(V,E)$ be a finite rooted tree with root $\rho$ and positive edge lengths. We also use $T$ for its continuous geometric realisation: each weighted edge is an interval of its specified length. Servers may occupy any point of the geometric tree, including
edge interiors. In our applications to general weighted trees, requests are
vertices and need not be leaves. Movement towards the root is
called upward, and movement away from the root is called downward. The distance $\dist(x,y)$ is the length of the unique path $[x,y]$.

For $z$ in the interior of an edge, deleting $z$ leaves two components. Let $T_z$ be the component not containing $\rho$, called the \emph{descendant side} of the cut at $z$. A point $x$ lies \emph{below} $z$ if $x\in T_z$. The lowest common ancestor $\lca(x,y)$ is the point of $[\rho,x]\cap[\rho,y]$ farthest from $\rho$, also when $x$ or $y$ is in an edge interior. The combinatorial depth of a vertex is the number
of original edges on its path from $\rho$.
The depth of $T$ is the maximum combinatorial
depth of any vertex.

\subsection{DoubleCoverage}
An online server at \(x\) is \emph{unobstructed} if no other online server
lies strictly between \(x\) and the request \(r\) on the unique path
\([x,r]\), and is \emph{obstructed} otherwise. If several servers occupy
the same point and no server lies strictly between that point and \(r\),
we designate one of them as unobstructed and regard the others as
obstructed.

\DC{} moves every currently unobstructed server toward the
request at unit speed. When a moving server reaches a vertex, it continues
along the unique path toward \(r\). The set of moving servers is updated whenever
the obstruction relation changes. The movement ends when a server reaches
\(r\). In the hard \(k\)-taxi problem, the serving taxi is then relocated
from \(r\) to \(s\) at no cost. Algorithm~\ref{alg:dc} describes both
settings.

\SetAlgoSkip{bigskip}
\begin{algorithm}[H]
\DontPrintSemicolon
\ForEach{request \(r\), or \((r,s)\) in the hard \(k\)-taxi setting}{
    \While{no online server is at \(r\)}{
        move every currently unobstructed online server toward \(r\) at unit
        speed until a server reaches \(r\) or the obstruction relation changes\;
    }
    \If{the request is \((r,s)\)}{
        relocate the serving taxi from \(r\) to \(s\) at zero cost\;
    }
}
\caption{\DC{} on a tree}
\label{alg:dc}
\end{algorithm}

\subsection{Potential Functions}
\label{subsec:potential-analysis}

Our notation
is adapted from~\cite{CP26}. For a finite sequence of $N$ requests, fix an optimal
offline solution. Let $\ALG_t$ and $\OPT_t$ denote
the costs incurred by the online and offline algorithms,
respectively, in serving request $t$. Thus
$\ALG=\sum_{t=1}^N\ALG_t$ and
$\OPT=\sum_{t=1}^N\OPT_t$.

Each analysis uses a nonnegative potential $\Phi(X,Y)$
whose initial value depends only on the tree and the
initial configurations. Let $\Phi_t$ be its value after
both algorithms have served request $t$, including the
free relocation in the taxi problem, and let $\Phi_0$
be its initial value. Write
$\Delta\Phi_t:=\Phi_t-\Phi_{t-1}$.
To prove $\mu$-competitiveness, it suffices to show that
\begin{equation}\label{eq:potential-criterion}
\ALG_t+\Delta\Phi_t\le\mu\OPT_t
\qquad\text{for every }t=1,\ldots,N.
\end{equation}
Indeed, summing these inequalities gives
\[
\begin{aligned}
\ALG+\Phi_N-\Phi_0
&=\sum_{t=1}^N\bigl(\ALG_t+\Delta\Phi_t\bigr)
\le\mu\sum_{t=1}^N\OPT_t
=\mu\OPT.
\end{aligned}
\]
Since $\Phi_N\ge0$, we obtain
$\ALG\le\mu\OPT+\Phi_0$, with an additive constant
independent of the request sequence.

For each request, we analyse offline movement first
and online movement second, holding the other
configuration fixed in each phase. In the taxi problem,
both serving taxis are then relocated from the pickup
to the destination at no cost. This ordering is used
only for the analysis and preserves both algorithms'
trajectories and costs.
We write $\Delta\Phi_t^{\mathrm{off}}$ and
$\Delta\Phi_t^{\mathrm{on}}$ for the potential changes
during offline and online movement, respectively,
and $\Delta\Phi_t^{\mathrm{rel}}$ for the change during
free relocation in the taxi problem.
The total change $\Delta\Phi_t$ is the sum of
the changes in these phases.

\section{Elementary Steps and Signed Cut Differences}
\label{sec:steps}
\subsection{Steps and Their Geometry}

Fix a request point $r$ and let $Y$ be the offline configuration
once the offline algorithm has a server at $r$. We hold $Y$
fixed while \DC{} moves towards $r$, with no movement needed
if an online server already occupies $r$.

\begin{definition}[Elementary step]\label{def:elementary-step}
An \emph{elementary step} is an interval $[a,b]$ of the online
movement, with duration $\delta=b-a>0$, such that the following
conditions hold during $(a,b)$:
\begin{enumerate}[(i)]
    \item The set of moving servers and the obstruction
    relation remain fixed.

    \item Each mover stays in the interior of one original
    edge and moves in a fixed direction.

    \item No mover passes an offline server.
\end{enumerate}
\end{definition}

\begin{lemma}\label{lem:step-decomposition}
Every DoubleCoverage movement phase on a finite weighted tree admits a
finite partition into elementary steps.
\end{lemma}

A proof is given in Appendix~\ref{app:events}.

We write $\Delta F$ for the value of $F$ after a step minus its value
before the step. Calculations on intervals strictly inside a step extend
to its endpoints by the continuity argument in
Appendix~\ref{app:continuity}.

Fix an elementary step with request point $r$.
Let $U$ and $B$ be the sets of upward and downward movers,
respectively, and put $u=|U|$. For each mover $i$, let $S_i$
be the segment from its initial to its final position, excluding
both endpoints. Choose $z_i\in S_i$ and, for $i\in U$, define
$A_i:=T_{z_i}$. All counts in these regions are evaluated
before the step.

\begin{lemma}[cf.~{\cite[Observation~7]{BCN23}}]
\label{lem:geometry}
There is at most one downward mover. The regions $A_i$,
$i\in U$, are pairwise disjoint and do not contain $r$.
If $B=\{j\}$, then $A:=T_{z_j}$ contains $r$ and every $A_i$.
The segments $S_i$ are also pairwise disjoint.
\end{lemma}

\begin{proof}
Every downward mover lies on $[\rho,r]$. If there were two,
the one nearer $r$ would obstruct the other.
Each upward mover travels out of a descendant region that
does not contain $r$. Two such regions cannot be nested,
since the higher mover would obstruct the lower one.
Thus the $A_i$ are pairwise disjoint and exclude $r$.

Suppose that $j$ moves downward. The path of any upward mover joins $[\rho,r]$
strictly below $j$, since otherwise $j$ would
obstruct it.
This joining point is an original vertex or $r$, so $j$ cannot reach it in the step interior.
Hence the joining point lies below every point of $S_j$, giving $A_i\subseteq A$ and $r\in A$.

The upward segments lie on separate branches and are therefore
pairwise disjoint. If there is a downward mover $j$, these branches
join $[\rho,r]$ below every point of $S_j$, so none of the upward
segments intersects $S_j$. Movers may meet at the beginning or end of the step,
but $S_i$ excludes the starting and finishing
positions of mover $i$.
\end{proof}

\subsection{Cut Differences}

For a region $S\subseteq T$, let
$D(S):=|X\cap S|-|Y\cap S|$ denote the number of online servers
in $S$ minus the number of offline servers in $S$, counting
servers with multiplicity. 

For a point $z$ in the interior
of an edge, define the \emph{signed cut difference} at $z$ by
\begin{equation}\label{eq:cut}
D(z):=D(T_z)=|X\cap T_z|-|Y\cap T_z|.
\end{equation}
Let $M(X,Y)$ be the minimum sum of distances in a perfect
matching between the online and offline servers. Below, integration over $T$ means summing the integrals over
its edges, each viewed as an interval of its given length.
Note that the values at vertices do not affect the integral.

 \begin{lemma}[Tree transportation identity]\label{lem:transport}
For any two $k$-point multisets on $T$,
\begin{equation}\label{eq:transport}
M(X,Y)=\int_T |D(z)|\,dz.
\end{equation}
\end{lemma}

Equation~\eqref{eq:transport} is a special case of the tree
Kantorovich--Rubinstein formula in~\cite{EM12}. We give a direct proof in Appendix~\ref{app:transport-proof}.

Relocating one online server and one offline server from
$r$ to $s$ changes both counts in $T_z$ by
$\mathbf{1}_{s\in T_z}-\mathbf{1}_{r\in T_z}$.
Their difference $D(z)$ therefore remains unchanged,
and so does $M$.

Moving the servers in either configuration a total distance
$C$ changes $M$ by at most $C$ in absolute value. For the
upper bound, use an optimal matching from before the
movement to pair the same servers afterwards and apply
the triangle inequality. Applying the argument in reverse,
starting from an optimal matching after the movement,
gives the lower bound.

Before the step, the cut difference is constant on each
segment $S_i$. Indeed, $S_i$ contains no vertex or offline
server by the definition of an elementary step, and any
online server initially in $S_i$ would obstruct mover $i$.
Let $q_i:=D(z_i)$ denote this constant value. After the
step, the cut difference on $S_i$ is $q_i-1$ if $i$ moves
upward and $q_i+1$ if $i$ moves downward. 

The corresponding
changes in $|D(z)|$ are given, for any integer $q$, by
\begin{align}
|q-1|-|q|
&=
\begin{cases}
-1,&q\ge1,\\
+1,&q\le0,
\end{cases}
\label{eq:up}\\
|q+1|-|q|
&=
\begin{cases}
+1,&q\ge0,\\
-1,&q\le-1.
\end{cases}
\label{eq:down}
\end{align}
By Lemma~\ref{lem:geometry}, the contributions add over
disjoint segments of length $\delta$. Other cut differences
are unchanged except possibly at segment endpoints, whose
values do not affect the integral.

\subsection{Proof of the Matching Estimate}

We show that the offline server at the request forces at least
one mover to contribute $-\delta$ to $\Delta M$ during an
elementary step of duration $\delta$.

\Needspace{9\baselineskip}
\begin{lemma}\label{lem:sign}
Assume that $Y$ is fixed and contains a server at $r$.
The following statements hold for cut differences evaluated
before an elementary step.
\begin{enumerate}[(a)]
    \item If $B=\emptyset$, then $D(A_i)\ge1$ for some $i\in U$.

    \item If $B=\{j\}$ and $A=T_{z_j}$, then either
    $D(A)\le-1$, or $D(A)\ge0$ and $D(A_i)\ge1$
    for some $i\in U$.
\end{enumerate}
\end{lemma}

\begin{proof}
\emph{Case $B=\emptyset$.}
We first show that every online server belongs to some
$A_i$ (see Figure~\ref{fig:cut-differences-no-downward}).
For a server initially at $x_p$, choose an online server
$i$ closest to $r$ on $[x_p,r]$. If several servers occupy
that position, choose the one that moves during the step.
No online server lies strictly between $x_i$ and $r$,
so $i$ is unobstructed. Since there is no downward mover,
$i$ moves upward. The point $z_i$ therefore lies on the
upward part of $[x_p,r]$, so $x_p\in T_{z_i}=A_i$.

Consequently, $R:=T\setminus\bigcup_{i\in U}A_i$ contains
no online server. By Lemma~\ref{lem:geometry}, none of
the $A_i$ contains $r$, so the offline server at $r$
belongs to $R$ and $D(R)\le-1$.
The regions $R$ and $A_i$, $i\in U$, partition $T$,
and $D(T)=0$ because both configurations contain
$k$ servers. Hence
\[
\sum_{i\in U}D(A_i)=D(T)-D(R)=-D(R)\ge1.
\]
Since the cut differences are integers, at least one
$i\in U$ satisfies $D(A_i)\ge1$, proving (a).

\emph{Case $B=\{j\}$.}
By Lemma~\ref{lem:geometry}, $A$ contains $r$ and every
$A_i$ (see Figure~\ref{fig:cut-differences-downward}).
If $D(A)\le-1$, the first alternative in (b) holds.
It remains to consider $D(A)\ge0$, since $D(A)$ is
an integer.

We show that every online server in $A$ belongs to
some $A_i$. For a server initially at $x_p\in A$,
the path $[x_p,r]$ lies entirely in $A$. Choose an
online server closest to $r$ on this path, using the
same rule for coincident servers as in (a).
This server is unobstructed and starts in $A$,
whereas the downward mover $j$ starts outside $A$.
It must therefore be an upward mover $i$, and the
same path argument as in (a) shows that $x_p\in A_i$.

Thus $R:=A\setminus\bigcup_{i\in U}A_i$ contains no
online server but contains the offline server at $r$,
so $D(R)\le-1$. Since the regions $R$ and $A_i$,
$i\in U$, partition $A$, we obtain
\[
\sum_{i\in U}D(A_i)=D(A)-D(R)\ge1.
\]
Therefore $D(A_i)\ge1$ for some $i\in U$, which
establishes the second alternative in (b).
\end{proof}

\begin{figure}[t]
\centering
\begin{tikzpicture}[
    x=0.72cm,
    y=0.72cm,
    online/.style={
        circle, draw=black, fill=black!75, text=white,
        minimum size=5.5mm, inner sep=0pt, font=\scriptsize
    },
    offline/.style={
        regular polygon, regular polygon sides=3,
        draw=black, fill=white, text=black,
        minimum size=7mm, inner sep=0pt, font=\scriptsize
    },
    treeedge/.style={
        draw=black!65, line width=0.7pt
    },
    region/.style={
        dashed, rounded corners=4pt,
        draw=black!55, fill=black!5
    },
    movement/.style={
        -{Stealth[length=2mm]}, line width=0.8pt
    },
    contribution/.style={
        font=\small\boldmath,
        fill=none, draw=none, inner sep=0pt
    },
    negative/.style={contribution,text=DCnegative},
    positive/.style={contribution,text=DCpositive},
    annotation/.style={
        font=\small, align=center, fill=none, draw=none
    }
]

\begin{scope}[on background layer]
    \draw[region] (-6.2,-5.4) rectangle (-1.9,-1.5);
    \draw[region] (-1.1,-5.4) rectangle (2.3,-1.5);
\end{scope}

\draw[treeedge] (0,0) -- (-4,-1) -- (-4,-2);
\draw[treeedge] (0,0) -- (0.6,-1) -- (0.6,-2);
\draw[treeedge] (0,0) -- (4.6,-1.5) -- (4.6,-4.7);
\draw[treeedge] (-4,-2) -- (-5.2,-4.6);
\draw[treeedge] (-4,-2) -- (-2.8,-4.6);
\draw[treeedge] (0.6,-2) -- (1.4,-4.6);

\node[
    circle, fill=black!65, inner sep=1.6pt,
    label=above:{$\rho$}
] at (0,0) {};

\node[online] at (-4,-2)     {$x_1$};
\node[online] at (0.6,-2)    {$x_2$};
\node[online] at (-5.2,-4.6) {$x_3$};

\node[offline] at (-2.8,-4.6) {$y_1$};
\node[offline] at (1.4,-4.6)  {$y_2$};
\node[offline] at (4.6,-4.7)  {$y_3$};

\node[annotation,anchor=north]
    at (4.6,-5.25) {$r$};

\fill (-4,-1.5) circle (1.4pt);
\fill (0.6,-1.5) circle (1.4pt);

\node[annotation,anchor=west]
    at (-3.85,-1.22) {$z_1$};
\node[annotation,anchor=west]
    at (0.75,-1.22) {$z_2$};

\draw[movement,draw=DCnegative]
    (-4.55,-2) -- (-4.55,-1.1)
    node[
        pos=0.2,left,negative,yshift=-0.5mm
    ] {$-\delta$};

\draw[movement,draw=DCpositive]
    (0.05,-2) -- (0.05,-1.1)
    node[
        pos=0.2,left,positive,yshift=-0.5mm
    ] {$+\delta$};

\node[annotation] at (-4.05,-6.15)
    {$A_1=T_{z_1}$\\$D(A_1)=1$};

\node[annotation] at (0.6,-6.15)
    {$A_2=T_{z_2}$\\$D(A_2)=0$};

\node[annotation,anchor=west,align=left]
    at (5.4,-2.3)
    {$R=T\setminus(A_1\cup A_2)$\\
     No online server\\
     $D(R)=-1$};

\end{tikzpicture}
\caption{Before an elementary step with $B=\emptyset$. The region
$R=T\setminus(A_1\cup A_2)$ contains the offline server at $r$ but no online
server. Arrow labels show the contributions to $\Delta M$.}
\label{fig:cut-differences-no-downward}
\end{figure}
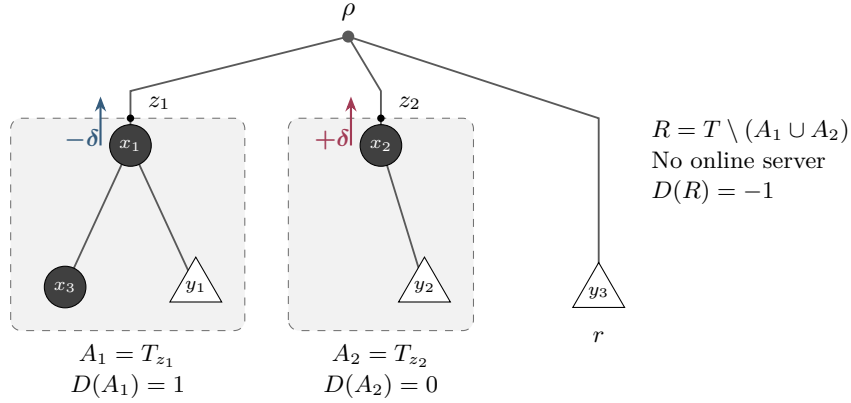

\begin{figure}[t]
\centering
\begin{tikzpicture}[
    x=0.68cm,
    y=0.68cm,
    online/.style={
        circle, draw=black, fill=black!75, text=white,
        minimum size=5.5mm, inner sep=0pt, font=\scriptsize
    },
    offline/.style={
        regular polygon, regular polygon sides=3,
        draw=black, fill=white, text=black,
        minimum size=7mm, inner sep=0pt, font=\scriptsize
    },
    treeedge/.style={
        draw=black!65, line width=0.7pt
    },
    outerregion/.style={
        dashed, rounded corners=4pt, draw=black!55
    },
    innerregion/.style={
        dashed, rounded corners=4pt,
        draw=black!45, fill=black!5
    },
    movement/.style={
        -{Stealth[length=2mm]}, line width=0.8pt
    },
    contribution/.style={
        font=\small\boldmath,
        fill=none, draw=none, inner sep=0pt
    },
    negative/.style={contribution,text=DCnegative},
    positive/.style={contribution,text=DCpositive},
    annotation/.style={
        font=\small, align=center, fill=none, draw=none
    }
]

\begin{scope}

\begin{scope}[on background layer]
    \draw[outerregion]
        (-4.1,-7.2) rectangle (3.7,-1.8);
    \draw[innerregion]
        (-3.7,-6.85) rectangle (-0.7,-3.6);
\end{scope}

\draw[treeedge] (0,0) -- (0,-3.15);
\draw[treeedge] (0,-0.65) -- (3.1,-1.15)
    coordinate[pos=0.62] (xTwo);
\draw[treeedge]
    (0,-3.15) -- (-2,-3.2) -- (-2,-5.25);
\draw[treeedge]
    (0,-3.15) -- (2,-3.2) -- (2,-5.2);

\node[
    circle, fill=black!65, inner sep=1.6pt,
    label=above:{$\rho$}
] at (0,0) {};

\node[online]  at (0,-1.25)   {$x_j$};
\node[online]  at (xTwo)      {$x_2$};
\node[offline] at (3.1,-1.15) {$y_3$};
\node[online]  at (-2,-4.1)   {$x_1$};
\node[offline] at (-2,-5.25)  {$y_1$};
\node[offline] at (2,-5.2)    {$y_2$};

\node[annotation,anchor=north]
    at (2,-5.7) {$r$};

\fill (0,-1.8) circle (1.3pt);
\fill (-2,-3.6) circle (1.3pt);

\node[annotation,anchor=east]
    at (-0.15,-1.65) {$z_j$};
\node[annotation,anchor=west]
    at (-1.85,-3.48) {$z_1$};

\draw[movement,draw=DCnegative]
    (0.5,-1.45) -- (0.5,-2.15)
    node[
        pos=0.75,right,negative,yshift=-1.2mm
    ] {$-\delta$};

\draw[movement,draw=DCpositive]
    (-2.5,-4.05) -- (-2.5,-3.35)
    node[
        pos=0.2,left,positive,yshift=-0.5mm
    ] {$+\delta$};

\node[
    annotation,inner sep=0pt,anchor=north west
] at (-3.65,-1.95)
    {$A=T_{z_j}$\\$D(A)=-1$};

\node[annotation,inner sep=0pt]
    at (-2,-6.1)
    {$A_1$\\$D(A_1)=0$};

\node[annotation,inner sep=0pt]
    at (0.75,-4.35) {$R$};

\node[font=\small] at (0,-7.7)
    {(a)\; $D(A)\le-1$};

\end{scope}

\begin{scope}[shift={(9,0)}]

\begin{scope}[on background layer]
    \draw[outerregion]
        (-4.4,-7.2) rectangle (3.7,-1.8);
    \draw[innerregion]
        (-4.05,-6.85) rectangle (-1.3,-3.6);
    \draw[innerregion]
        (-0.9,-6.85) rectangle (1.9,-3.6);
\end{scope}

\draw[treeedge] (0,0) -- (0,-3.15);
\draw[treeedge] (0,-0.65) -- (2.6,-1.1);
\draw[treeedge]
    (0,-3.15) -- (-2.5,-3.2) -- (-2.5,-4.1);
\draw[treeedge]
    (0,-3.15) -- (0.5,-3.2) -- (0.5,-5.25);
\draw[treeedge]
    (0,-3.15) -- (2.8,-3.2) -- (2.8,-5.2);

\node[
    circle, fill=black!65, inner sep=1.6pt,
    label=above:{$\rho$}
] at (0,0) {};

\node[online]  at (0,-1.25)   {$x_j$};
\node[offline] at (2.6,-1.1)  {$y_3$};
\node[online]  at (-2.5,-4.1) {$x_1$};
\node[online]  at (0.5,-4.1)  {$x_2$};
\node[offline] at (0.5,-5.25) {$y_1$};
\node[offline] at (2.8,-5.2)  {$y_2$};

\node[annotation,anchor=north]
    at (2.8,-5.7) {$r$};

\fill (0,-1.8) circle (1.3pt);
\fill (-2.5,-3.6) circle (1.3pt);
\fill (0.5,-3.6) circle (1.3pt);

\node[annotation,anchor=east]
    at (-0.15,-1.65) {$z_j$};
\node[annotation,anchor=west]
    at (-2.35,-3.48) {$z_1$};
\node[annotation,anchor=west]
    at (0.65,-3.48) {$z_2$};

\draw[movement,draw=DCpositive]
    (0.5,-1.45) -- (0.5,-2.15)
    node[
        pos=0.75,right,positive,yshift=-1.2mm
    ] {$+\delta$};

\draw[movement,draw=DCnegative]
    (-3,-4.05) -- (-3,-3.35)
    node[
        pos=0.2,left,negative,yshift=-0.5mm
    ] {$-\delta$};

\draw[movement,draw=DCpositive]
    (0,-4.05) -- (0,-3.35)
    node[
        pos=0.2,left,positive,yshift=-0.5mm
    ] {$+\delta$};

\node[
    annotation,inner sep=0pt,anchor=north west
] at (-3.85,-1.95)
    {$A=T_{z_j}$\\$D(A)=0$};

\node[annotation,inner sep=0pt]
    at (-2.675,-6.1)
    {$A_1$\\$D(A_1)=1$};

\node[annotation,inner sep=0pt]
    at (0.5,-6.1)
    {$A_2$\\$D(A_2)=0$};

\node[annotation,inner sep=0pt]
    at (2.25,-4.35) {$R$};

\node[font=\small] at (0,-7.7)
    {(b)\; $D(A)\ge0$};

\end{scope}

\end{tikzpicture}
\caption{Before an elementary step with $B=\{j\}$. In both panels,
$R=A\setminus\bigcup_{i\in U}A_i$ contains no online server and has
$D(R)=-1$. Arrow labels show the contributions to $\Delta M$.}
\label{fig:cut-differences-downward}
\end{figure}
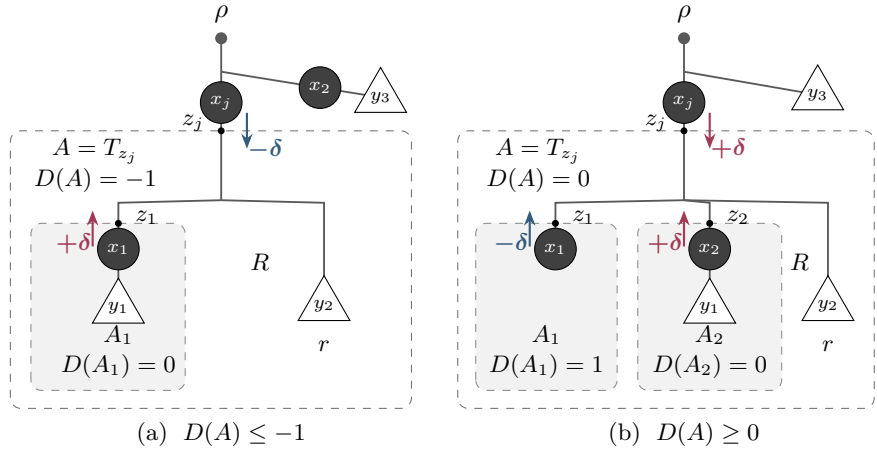

\Needspace{10\baselineskip}
The estimate below is also implied by the tree analysis
of Bansal et al.~\cite{BEJKP18}.

\Needspace{8\baselineskip}
\begin{corollary}[Matching estimate]
\label{cor:matching}
Assume that $Y$ is fixed and contains a server at $r$.
For an elementary step of duration $\delta$ with
$a=u+|B|$ movers,
\begin{equation}\label{eq:matching-step}
\Delta M\le(a-2)\delta
=
\begin{cases}
(u-2)\delta,&B=\emptyset,\\
(u-1)\delta,&B=\{j\}.
\end{cases}
\end{equation}
\end{corollary}
\begin{proof}
Suppose first that $B=\emptyset$. By Lemma~\ref{lem:sign},
some upward mover $i$ has $D(A_i)\ge1$, so its contribution
to $\Delta M$ is $-\delta$ by~\eqref{eq:up}.

Now suppose that $B=\{j\}$, and put $A=T_{z_j}$.
If $D(A)\le-1$, the downward mover $j$ contributes
$-\delta$ by~\eqref{eq:down}. Otherwise,
Lemma~\ref{lem:sign} shows that some upward mover $i$
satisfies $D(A_i)\ge1$. Its contribution to $\Delta M$
is therefore $-\delta$ by~\eqref{eq:up}.

Thus, in both cases, at least one mover contributes $-\delta$,
while every other mover contributes at most $\delta$.
With $a=u+|B|$ movers, this gives
$\Delta M\le-\delta+(a-1)\delta=(a-2)\delta$.
\end{proof}

\FloatBarrier
\section{The \texorpdfstring{\(k\)}{k}-Server Problem on Trees}
\label{sec:server}
 We apply the matching estimate to give an alternative proof
of the $k$-competitiveness of \DC{} for the $k$-server problem
on weighted trees. The argument uses the potential of
Chrobak and Larmore~\cite{CL91} and is similar to the analysis in~\cite{BEJKP18}.

For the online configuration $X=\{x_1,\ldots,x_k\}$,
let $P(X)$ be the sum of distances over all unordered
pairs of online servers. We use the potential
\begin{equation}\label{eq:server-potential}
\Phi_{\mathrm{srv}}(X,Y)=P(X)+kM(X,Y),\qquad
P(X)=\sum_{1\le a<b\le k}\dist(x_a,x_b).
\end{equation}
Both terms are nonnegative. To analyse online movement,
we combine Corollary~\ref{cor:matching} with the following
formula for the change in $P$.

\begin{lemma}\label{lem:pair}
Let $a=u+|B|$ be the total number of movers in an
elementary step of duration $\delta$. Then
\begin{equation}\label{eq:pair}
\Delta P=[2k-(k+1)a]\delta.
\end{equation}
\end{lemma}

\begin{proof}
Let $\Theta(X):=\sum_{x\in X}\dist(\rho,x)$ be the total
distance of the online servers from the root, and set
$L(X):=-\sum_{a<b}\dist(\rho,\lca(x_a,x_b))$.
For each pair of servers,
\[
\dist(x_a,x_b)
=\dist(\rho,x_a)+\dist(\rho,x_b)
-2\dist(\rho,\lca(x_a,x_b)).
\]
Since each server appears in $k-1$ pairs, summing gives
$P=(k-1)\Theta+2L$. We first consider server positions
during the interior of the step.

\emph{Case $B=\emptyset$.}
Here $\Delta\Theta=-u\delta$.
The proof of Lemma~\ref{lem:sign}(a) shows that every
stationary server lies below exactly one upward mover.
Their lowest common ancestor is that mover, whose
distance from the root decreases by $\delta$, so
each such pair contributes $\delta$ to $\Delta L$.

All other pairs involving a mover have fixed lowest
common ancestors because their servers lie on separate
branches, while pairs of stationary servers are unchanged.
There are $k-u$ stationary servers, so
$\Delta L=(k-u)\delta$ and
\[
\Delta P=-(k-1)u\delta+2(k-u)\delta
       =[2k-(k+1)u]\delta.
\]

\emph{Case $B=\{j\}$.}
Here $\Delta\Theta=(1-u)\delta$.
Every upward mover lies below $j$, so its lowest common
ancestor with $j$ is $j$. Since $j$ moves away from
the root by $\delta$, these $u$ pairs contribute
$-u\delta$ to $\Delta L$.

Each stationary server below $j$ lies below exactly
one upward mover. Its pair with that mover contributes
$\delta$ to $\Delta L$, whereas its pair with $j$
contributes $-\delta$, so the two contributions cancel.
Pairs involving a stationary server outside the subtree
below $j$ have fixed lowest common ancestors, as do
all remaining pairs. Hence $\Delta L=-u\delta$ and
\[
\Delta P=(k-1)(1-u)\delta-2u\delta
       =[(k-1)-(k+1)u]\delta.
\]

Since $a=u$ in the first case and $a=u+1$ in the
second, both expressions equal $[2k-(k+1)a]\delta$.
Appendix~\ref{app:continuity} shows that $P$ is continuous
as the servers move, including when they meet or reach
vertices. Taking limits from within the step therefore
gives the stated identity for $\Delta P$ over the whole step.
\end{proof}

\Needspace{10\baselineskip}
The following proof uses the notation from Section~\ref{subsec:potential-analysis}.

\begin{theorem}[{\cite{CL91}}]
\label{thm:server}
DoubleCoverage is $k$-competitive for the $k$-server
problem on finite weighted trees.
\end{theorem}

\begin{proof}
Using $\Phi=\Phi_{\mathrm{srv}}=P+kM$, it suffices to
prove $\ALG_t+\Delta\Phi_t\le k\OPT_t$ for every
request $t$. Fix request $t$ at $r$.

\emph{Offline movement.}
The online configuration remains fixed, so $P$ is
unchanged. Keeping the pairs of a minimum matching
fixed as the offline servers move shows that
$\Delta M\le\OPT_t$. Hence
$\Delta\Phi_t^{\mathrm{off}}=k\Delta M\le k\OPT_t$.

\emph{Online movement.}
The offline configuration now contains a server at
$r$ and remains fixed while \DC{} serves the request.
Consider an elementary step of duration $\delta$
with $a$ moving servers, and write $\Delta\Phi$ for
its potential change. Its movement cost is
$C_{\mathrm{DC}}=a\delta$.
Lemma~\ref{lem:pair} and Corollary~\ref{cor:matching}
give, respectively,
\[
\Delta P=[2k-(k+1)a]\delta,
\qquad
\Delta M\le(a-2)\delta.
\]
Substituting these bounds into the potential gives
\[
\Delta\Phi
=\Delta P+k\Delta M
\le \bigl[2k-(k+1)a+k(a-2)\bigr]\delta
=-a\delta
=-C_{\mathrm{DC}}.
\]
The step costs sum to $\ALG_t$, and the potential
changes sum to $\Delta\Phi_t^{\mathrm{on}}$, so
summing these inequalities gives
$\ALG_t+\Delta\Phi_t^{\mathrm{on}}\le0$.
Combining this with the offline bound yields
\[
\ALG_t+\Delta\Phi_t
=\Delta\Phi_t^{\mathrm{off}}
 +\bigl(\ALG_t+\Delta\Phi_t^{\mathrm{on}}\bigr)
\le k\OPT_t,
\]
as desired.
\end{proof}

\section{The \texorpdfstring{\(k\)}{k}-Taxi Problem on HSTs}
\label{sec:hst}

Our potential is a weighted sum of a normalized
version of the rank potential
from~\cite[Section~4.3]{BCN23}, the total distance
of the online taxis from the root, and the minimum
matching distance between the online and offline
configurations.

\subsection{HSTs}
\label{subsec:hst-definition}

Introduced by Bartal~\cite{Bar96}, hierarchically
well-separated trees are a class of rooted trees
used to approximate finite metrics through randomized
embeddings with logarithmic expected
distortion~\cite{FRT04}.

For \(\alpha > 1\), an
\(\alpha\)-\emph{hierarchically well-separated tree} (\(\alpha\)-HST) is a
rooted weighted tree \(T\) such that all leaves have the same combinatorial
depth \(d\), and edge lengths decrease by a factor of at least \(\alpha\)
when moving from the root toward a leaf.  The metric space consists of the leaves, so all pickups and destinations
are leaves.

Throughout this section, we use the layered HST model
of~\cite{BCN23}, in which all edges joining the same
two consecutive depths have equal length.
Every path from the root to a leaf therefore has
the same length $W$. For $x\in T$, define its
\emph{weighted height} by $h(x):=W-\dist(\rho,x)$,
which is the distance from $x$ to any descendant leaf.

A vertex at depth $d-\ell$ has
\emph{combinatorial height} $\ell$.
Write $\alpha_\ell$ for the common weighted height
of these vertices, so
$0=\alpha_0<\alpha_1<\cdots<\alpha_d=W$.
For $\ell=0,\ldots,d-1$, \emph{layer $\ell$}
consists of the interiors of the edges joining
combinatorial heights $\ell$ and $\ell+1$.
Each such edge has length
$\alpha_{\ell+1}-\alpha_\ell$, and a point $x$
lies in layer $\ell$ exactly when
$\alpha_\ell<h(x)<\alpha_{\ell+1}$.

\subsection{The HST Rank Potential}

We order the online taxis by weighted height,
$h_0\le h_1\le\cdots\le h_{k-1}$, with ranks $0,\ldots,k-1$.
Ties may be ordered arbitrarily.

For nonnegative integers \(m,\ell\), define
\[
    c_{m,\ell}
    :=\sum_{q=1}^{\min\{m,\ell\}}\binom{m}{q}.
\]

In particular, $c_{0,\ell}=c_{m,0}=0$.

For $\ell=0,\ldots,d-1$, let $\chi_\ell(h)$ be the length of the part of
$[\alpha_\ell,\alpha_{\ell+1}]$ below height $h$:
\begin{equation}
\label{eq:chi}
\chi_\ell(h):=
\begin{cases}
    0, & h\le \alpha_\ell,\\
    h-\alpha_\ell, & \alpha_\ell<h<\alpha_{\ell+1},\\
    \alpha_{\ell+1}-\alpha_\ell, & h\ge \alpha_{\ell+1}.
\end{cases}
\end{equation}

We define
\[
    \Psi(X)
    :=\sum_{i=0}^{k-1}\sum_{\ell=0}^{d-1}
      c_{i,\ell}\chi_\ell(h_i).
\]
The height term used in~\cite[Section~4.3]{BCN23}
is $\max\{\alpha_\ell,\min\{h,\alpha_{\ell+1}\}\}
=\chi_\ell(h)+\alpha_\ell$.
Thus our $\Psi$ differs from their rank potential
only by an additive constant.
Every term in $\Psi$ is nonnegative, and
$\chi_\ell(0)=0$ for every $\ell$, so $\Psi=0$
when all taxis are at leaves.

\begin{lemma}[{\cite[Lemma~11]{BCN23}}]
\label{lem:coefficient-recurrence}
For $m\ge 0$ and $\ell\ge 1$,
\[
c_{m,\ell}
=
m+\sum_{i=0}^{m-1} c_{i,\ell-1}.
\]
\end{lemma}

\begin{proof}
The case $m=0$ is immediate, with an empty sum. For $m\ge1$, count the
nonempty subsets of $\{0,\ldots,m-1\}$ of size at most $\ell$ by their
largest element $i$. The remaining elements can be chosen in
$c_{i,\ell-1}+1$ ways, including the empty set, giving
\[
c_{m,\ell}=\sum_{i=0}^{m-1}(c_{i,\ell-1}+1)
=m+\sum_{i=0}^{m-1}c_{i,\ell-1}.
\qedhere
\]
\end{proof}

We subdivide each elementary step whenever two
taxis exchange their order by height and call
the resulting intervals \emph{rank-refined steps}.
Throughout an elementary step, each labelled
taxi's height increases at unit speed, decreases
at unit speed, or remains constant.
Two such heights therefore either agree throughout
the step or coincide at most once, so the
refinement is finite.
We fix an arbitrary order among taxis whose
heights agree throughout a step, so every taxi
has a fixed rank in the step interior.
Appendix~\ref{app:continuity} shows that $\Psi$
is continuous, including when taxis exchange ranks
or reach vertices. Taking limits at the step endpoints
therefore extends the bounds proved in the step
interior to the whole step.

\begin{lemma}
\label{lem:rank-potential-change}
Consider a rank-refined step of duration \(\delta\).  Let \(U\) be the set
of taxis moving upward, let \(u=|U|\), and let \(B\) be the set of
taxis moving downward.  Then
\[
    \Delta\Psi\le
    \begin{cases}
        (c_{k,d}-u)\delta, & B=\varnothing,\\[2mm]
        -u\delta,         & B=\{j\}.
    \end{cases}
\]
\end{lemma}

\begin{proof}
Label each taxi by its rank during the step, and let
$\ell_i$ be the layer containing moving taxi $i$.
For this taxi, only $\chi_{\ell_i}(h_i)$ changes,
increasing by $\delta$ for upward movement and decreasing
by $\delta$ for downward movement. Its contribution to
$\Delta\Psi$ is therefore $c_{i,\ell_i}\delta$ or
$-c_{i,\ell_i}\delta$, respectively, while stationary
taxis contribute zero.

\emph{Case $B=\emptyset$.}
Here $u\le k$ and $\ell_i\le d-1$ for every mover.
The coefficients are nonnegative and nondecreasing
in their second index, so
\[
\begin{aligned}
u+\frac{\Delta\Psi}{\delta}
&=u+\sum_{i\in U}c_{i,\ell_i}\\
&\le k+\sum_{i=0}^{k-1}c_{i,d-1}
=c_{k,d},
\end{aligned}
\]
where the last equality follows from
Lemma~\ref{lem:coefficient-recurrence}.
Rearranging gives $\Delta\Psi\le(c_{k,d}-u)\delta$.

\emph{Case $B=\{j\}$.}
If $\ell_j=0$, taxi $j$ lies on the edge leading
to the request leaf. No branch joins the path
between this taxi and the request, so $u=0$ and
$\Delta\Psi=-c_{j,0}\delta=0$.

It remains to consider $\ell_j\ge1$.
For each upward mover $i$, let $v_i$ be the vertex
where its path to the request joins $[\rho,r]$.
By Lemma~\ref{lem:geometry}, $v_i$ lies below
taxi $j$, while taxi $i$ lies below $v_i$ throughout
the interior of the step. Every vertex below a taxi
in layer $\ell_j$ has height at most $\alpha_{\ell_j}$,
so $h_i<h(v_i)\le\alpha_{\ell_j}<h_j$.
It follows that $i<j$ and $\ell_i\le\ell_j-1$
for every $i\in U$, which also implies $u\le j$.
Using these inequalities and the nonnegativity
and monotonicity of the coefficients, we obtain
\[
\begin{aligned}
u+\frac{\Delta\Psi}{\delta}
&=u-c_{j,\ell_j}
  +\sum_{i\in U}c_{i,\ell_i}\\
&\le j-c_{j,\ell_j}
  +\sum_{i=0}^{j-1}c_{i,\ell_j-1}\\
&=0,
\end{aligned}
\]
where the last equality is
Lemma~\ref{lem:coefficient-recurrence} with
$m=j$ and $\ell=\ell_j$.
Thus $\Delta\Psi\le-u\delta$.
\end{proof}

\subsection{Finalizing the Analysis}
\label{subsec:hst-online}

We combine the rank-potential bound with the matching estimate
to obtain a potential whose decrease pays for online movement.
With $\Theta(X)=\sum_{x\in X}\dist(\rho,x)$ as in
Section~\ref{sec:server}, we define the potential
\begin{equation}
\label{eq:hst-potential}
    \Phi(X,Y)
    :=(c_{k,d}-1)\Theta(X)+2\Psi(X)+c_{k,d}M(X,Y).
\end{equation}
Clearly all three terms are nonnegative.
We use the notation from Section~\ref{subsec:potential-analysis} below.
\begin{lemma}
\label{lem:hst-online}
Fix request $t$ with pickup $r$.
Suppose that the offline configuration $Y$ remains fixed
and contains a taxi at $r$ while \DC{} serves the pickup.
Then
\[
    \ALG_t+\Delta\Phi_t^{\mathrm{on}}\le0.
\]
\end{lemma}

\begin{proof}
Consider a rank-refined step of duration $\delta$,
and write $C_{\mathrm{DC}}$ for its movement cost
and $\Delta\Phi$ for its potential change.
We use Lemma~\ref{lem:rank-potential-change} and
Corollary~\ref{cor:matching} to bound $\Delta\Psi$
and $\Delta M$, respectively.

\emph{Case $B=\emptyset$.}
Here $C_{\mathrm{DC}}=u\delta$ and
$\Delta\Theta=-u\delta$, while the two estimates imply
$\Delta\Psi\le(c_{k,d}-u)\delta$ and
$\Delta M\le(u-2)\delta$.
Substituting into~\eqref{eq:hst-potential} yields
\[
C_{\mathrm{DC}}+\Delta\Phi
\le
\bigl[u-(c_{k,d}-1)u
+2(c_{k,d}-u)+c_{k,d}(u-2)\bigr]\delta
=0.
\]

\emph{Case $B=\{j\}$.}
Here $C_{\mathrm{DC}}=(u+1)\delta$ and
$\Delta\Theta=(1-u)\delta$.
The corresponding bounds are $\Delta\Psi\le-u\delta$
and $\Delta M\le(u-1)\delta$, so
\[
C_{\mathrm{DC}}+\Delta\Phi
\le
\bigl[u+1+(c_{k,d}-1)(1-u)
-2u+c_{k,d}(u-1)\bigr]\delta
=0.
\]

The step costs sum to $\ALG_t$ and the potential
changes to $\Delta\Phi_t^{\mathrm{on}}$, so adding
these inequalities proves the claim.
\end{proof}

\begin{theorem}[{\cite{BCN23}}]
\label{thm:hst}
On a depth-\(d\) HST, \DC{} is
\(c_{k,d}\)-competitive for the hard \(k\)-taxi problem, where
\[
    c_{k,d}
    = \sum_{h=1}^{\min\{k,d\}}\binom{k}{h}.
\]
In particular, if \(d\ge k\), then \(c_{k,d}=2^k-1\).
\end{theorem}

\begin{proof}
By the criterion in Section~\ref{subsec:potential-analysis},
it suffices to prove
$\ALG_t+\Delta\Phi_t\le c_{k,d}\OPT_t$
for every request $t$.
Fix request $t$, with pickup $r$ and destination $s$,
and consider its three phases.

\emph{Offline movement.}
The online configuration remains fixed, so $\Theta$
and $\Psi$ are unchanged. The bound $\Delta M\le\OPT_t$
therefore implies
$\Delta\Phi_t^{\mathrm{off}}
=c_{k,d}\Delta M\le c_{k,d}\OPT_t$.

\emph{Online movement.}
The offline configuration contains a taxi at $r$
and remains fixed while \DC{} serves the pickup.
Lemma~\ref{lem:hst-online} then yields
$\ALG_t+\Delta\Phi_t^{\mathrm{on}}\le0$.

\emph{Free relocation.}
Both configurations relocate their serving taxi
from $r$ to $s$, preserving every cut difference
and hence $M$.
Both leaves lie at distance $W$ from the root
and have height zero, so $\Theta$ and the multiset
of online taxi heights remain unchanged.
This also preserves $\Psi$, and therefore
$\Delta\Phi_t^{\mathrm{rel}}=0$.

Adding the bounds for the three phases, we obtain
\[
\begin{aligned}
\ALG_t+\Delta\Phi_t
&=\Delta\Phi_t^{\mathrm{off}}
  +\bigl(\ALG_t+\Delta\Phi_t^{\mathrm{on}}\bigr)
  +\Delta\Phi_t^{\mathrm{rel}}\\
&\le c_{k,d}\OPT_t,
\end{aligned}
\]
as required.
\end{proof}

\section{The \texorpdfstring{\(k\)}{k}-Taxi Problem on General Weighted Trees}
\label{sec:tree}

On a general weighted tree, the pickup and destination
may lie at different distances from the root, so
$\Theta$ need not be preserved by free relocation.
We therefore use a potential that depends only on
the cut differences, which remain unchanged in this phase.

Let $k\ge2$, and let $T$ be a finite rooted weighted
tree of combinatorial depth $d\ge1$.
Pickups and destinations are vertices of the original
tree and may be internal vertices.
For a point $z$ in the interior of an original edge,
let $\ell(z)$ be the combinatorial depth of its
endpoint farther from the root.
We also call $\ell(z)$ the depth of the edge,
so $\ell(z)\in\{1,\ldots,d\}$.
This index refers to the original tree and is unchanged
by any subdivision used in the analysis.

\subsection{Coefficients and the Potential Function}

We use the coefficients $m_i<0<M_i$, indexed by edge depth,
from~\cite[Section~5.1]{BCN23}.
At each point of an edge of depth $i$, the potential
assigns weight $M_i$ to each excess online taxi below
that point and weight $-m_i$ to each excess offline taxi.
For $i=1,\ldots,d$, define

\begin{align}
m_i
&=
\begin{cases}
-2(d-i)-1, & k=2,\\[2pt]
\dfrac{-2(k-1)^{d-i+1}+k}{k-2}, & k\ge3,
\end{cases}
\label{eq:mi}\\
M_i
&=
\begin{cases}
2(d+i)-1, & k=2,\\[2pt]
\dfrac{2k(k-1)^d-2(k-1)^{d-i+1}-k}{k-2}, & k\ge3.
\end{cases}
\label{eq:Mi}
\end{align}
Here $M_i$ denotes a coefficient, whereas $M(X,Y)$
denotes the minimum matching distance.
The following properties, proved
in~\cite[Section~5.1]{BCN23}, will be used to bound
the potential changes during online and offline movement.

\begin{lemma}[{\cite{BCN23}}]
\label{lem:tree-coeff}
Both coefficient sequences are strictly increasing, with
$m_d=-1$ and $M_1>0$, and the difference $M_i-m_i$ is independent
of $i$. For every integer $1\le t\le k$,
\begin{align}
M_1+(t-1)m_1
&\ge t,
\label{eq:coeff-up}\\
(t-1)m_{i+1}-m_i
&\ge t
\qquad (1\le i<d).
\label{eq:coeff-down}
\end{align}
Moreover,
\begin{equation}
\label{eq:tree-factor}
\max_{1\le i\le d}\{M_i,|m_i|\}
=M_d
=
\begin{cases}
4d-1, & k=2,\\[2pt]
\dfrac{2k(k-1)^d-3k+2}{k-2}, & k\ge3.
\end{cases}
\end{equation}
\end{lemma}

For each depth $i$, let $\varphi_i(q)$ be the weighted value
assigned to a cut difference $q$:
\[
\varphi_i(q)
=
\begin{cases}
M_iq, & q\ge0,\\
m_iq, & q<0,
\end{cases}
\qquad q\in\mathbb{Z}.
\]
Integrating these weighted cut differences over the tree gives
our potential:
\begin{equation}
\label{eq:tree-potential}
\Phi_{\mathrm{tree}}(X,Y)
=
\int_T \varphi_{\ell(z)}(D(z))\,dz.
\end{equation}
The inequalities $M_i>0$ and $m_i<0$ imply
$\varphi_i(q)\ge0$ for every integer $q$, so the potential
is nonnegative. Relocating one taxi from $r$ to $s$ in each
configuration preserves every cut difference and therefore
leaves the potential unchanged.

\begin{samepage}
During an elementary step, an online taxi moving upward
changes the cut difference from $q$ to $q-1$ at each
point it passes, while a taxi moving downward changes
it from $q$ to $q+1$.
On an edge of depth $i$, the contribution to
$\Delta\Phi_{\mathrm{tree}}$ per unit distance travelled
is therefore $\varphi_i(q-1)-\varphi_i(q)$ or
$\varphi_i(q+1)-\varphi_i(q)$, respectively:
\begin{equation}
\label{eq:tree-slopes}
\begin{array}{c|c|c}
\text{direction}
& \text{cut difference before the step}
& \text{contribution to $\Delta\Phi_{\mathrm{tree}}$ per unit length}\\
\hline
\text{upward}   & q\ge1  & -M_i\\
\text{upward}   & q\le0  & -m_i\\
\text{downward} & q\ge0  & M_i\\
\text{downward} & q\le-1 & m_i
\end{array}
\end{equation}
\end{samepage}
By Lemma~\ref{lem:geometry}, the segments $S_i$
are pairwise disjoint, so we may sum the contributions.

Moving one taxi in either configuration across a cut changes \(D(z)\) by
\(1\) or \(-1\). By Lemma~\ref{lem:tree-coeff}, \(M_i\le M_d\) and
\(|m_i|\le M_d\) for every \(i\), so the integrand changes by at most
\(M_d\) in absolute value at each crossed cut. Integrating along the movement
paths shows that, if one configuration moves a total distance \(C\) while
the other is fixed,
\begin{equation}
\label{eq:tree-motion-bound}
|\Delta\Phi_{\mathrm{tree}}|\le M_d C.
\end{equation}

\subsection{Finalizing the Analysis}

Throughout this subsection, write $\Phi=\Phi_{\mathrm{tree}}$
and use the notation from
Section~\ref{subsec:potential-analysis}.

\begin{lemma}
\label{lem:tree-online}
Fix request $t$ with pickup $r$.
Suppose that the offline configuration $Y$ remains fixed
and contains a taxi at $r$ while \DC{} serves the pickup.
Then
\[
    \ALG_t+\Delta\Phi_t^{\mathrm{on}}\le0.
\]
\end{lemma}

\begin{proof}
Consider an elementary step of duration $\delta$,
and write $C_{\mathrm{DC}}$ for its movement cost
and $\Delta\Phi$ for its potential change.
We use the notation of Section~\ref{sec:steps},
with $q_i=D(z_i)$ evaluated before the step.
Each mover travels distance $\delta$, so its contribution
to $\Delta\Phi$ is $\delta$ times the corresponding
entry in~\eqref{eq:tree-slopes}.

\emph{Case $B=\emptyset$.}
Here $C_{\mathrm{DC}}=u\delta$ and $1\le u\le k$.
By Lemma~\ref{lem:sign}, some upward mover $i$
has $q_i\ge1$.
By~\eqref{eq:tree-slopes} and the monotonicity of the
coefficients, this mover contributes at most
$-M_1\delta$, while every other mover contributes
at most $-m_1\delta$.
Therefore
\[
C_{\mathrm{DC}}+\Delta\Phi
\le\bigl[u-M_1-(u-1)m_1\bigr]\delta
\le0,
\]
where the last inequality follows from
\eqref{eq:coeff-up} with $t=u$.

\emph{Case $B=\{j\}$.}
Here $C_{\mathrm{DC}}=(u+1)\delta$.
Let $h$ be the depth of the edge containing
the downward mover $j$, and let $v$ be its
endpoint farther from the root.

By Lemma~\ref{lem:geometry}, the path of each upward
mover joins $[\rho,r]$ at a point that remains below
taxi $j$ throughout the interior of the step.
This joining point is an original vertex,
either a branching vertex or the request $r$,
and is therefore $v$ or a descendant of $v$.
The upward mover lies strictly below this point,
so its edge has depth at least $h+1$.

If $h=d$, there are no deeper edges and hence
no upward movers.
Lemma~\ref{lem:sign} then implies $q_j\le-1$,
so $\Delta\Phi=m_d\delta=-\delta=-C_{\mathrm{DC}}$.

For $h<d$, we consider the sign of $q_j$.
If $q_j\le-1$, the downward mover contributes
$m_h\delta$. By the depth comparison and the
monotonicity of the coefficients, each upward mover
contributes at most $-m_{h+1}\delta$.
Thus
\[
    \Delta\Phi\le(m_h-u m_{h+1})\delta.
\]
If $q_j\ge0$, Lemma~\ref{lem:sign} shows that there is an
upward mover $i$ with $q_i\ge1$.
This taxi contributes at most $-M_{h+1}\delta$,
while the downward mover contributes $M_h\delta$
and every other upward mover contributes at most
$-m_{h+1}\delta$.
Consequently,
\[
\begin{aligned}
\Delta\Phi
&\le
\bigl[M_h-M_{h+1}-(u-1)m_{h+1}\bigr]\delta\\
&=(m_h-u m_{h+1})\delta,
\end{aligned}
\]
where the equality follows from the fact that
$M_i-m_i$ is independent of $i$.
For either sign of $q_j$, applying
\eqref{eq:coeff-down} with $t=u+1\le k$ yields
\[
C_{\mathrm{DC}}+\Delta\Phi
\le\bigl[u+1+m_h-u m_{h+1}\bigr]\delta
\le0.
\]

Summing $C_{\mathrm{DC}}+\Delta\Phi\le0$ over the
elementary steps gives
$\ALG_t+\Delta\Phi_t^{\mathrm{on}}\le0$.
\end{proof}

\begin{theorem}[{\cite{BCN23}}]
\label{thm:tree}
For $k\ge2$, DoubleCoverage is $M_d$-competitive for hard
$k$-taxi on a finite rooted weighted tree of combinatorial
depth $d\ge1$, with pickups and destinations at its vertices.
Here $M_d$ is given by~\eqref{eq:tree-factor}
and is $O(k^d)$ for fixed $d$.
\end{theorem}
\begin{proof}
It suffices to show
$\ALG_t+\Delta\Phi_t\le M_d\OPT_t$
for every request $t$.
Fix request $t$, with pickup $r$ and destination $s$,
and consider its three phases.

\emph{Offline movement.}
The online configuration remains fixed while the
offline solution serves the pickup at cost $\OPT_t$.
By~\eqref{eq:tree-motion-bound},
$\Delta\Phi_t^{\mathrm{off}}\le M_d\OPT_t$.

\emph{Online movement.}
The offline configuration contains a taxi at $r$
and remains fixed while \DC{} serves the pickup.
By Lemma~\ref{lem:tree-online},
$\ALG_t+\Delta\Phi_t^{\mathrm{on}}\le0$.

\emph{Free relocation.}
Both configurations relocate their serving taxi
from $r$ to $s$, preserving every cut difference
and hence $\Phi$.
Thus $\Delta\Phi_t^{\mathrm{rel}}=0$.

Adding the bounds for the three phases, we obtain
\[
\begin{aligned}
\ALG_t+\Delta\Phi_t
&=\Delta\Phi_t^{\mathrm{off}}
  +\bigl(\ALG_t+\Delta\Phi_t^{\mathrm{on}}\bigr)
  +\Delta\Phi_t^{\mathrm{rel}}\\
&\le M_d\OPT_t,
\end{aligned}
\]
as desired.
\end{proof}

\paragraph*{Acknowledgments.}
The author thanks Allan Borodin for his encouragement, insight, guidance, and helpful comments on this manuscript.

\paragraph*{Generative AI disclosure.}
Generative AI tools were used extensively during the development and preparation of this manuscript. The initial proof idea was suggested by ChatGPT 5.6 Sol while the author was seeking assistance with a different problem. ChatGPT 5.6 Sol and Astra were subsequently used to explore, develop, and formalize the argument and to assist with drafting.

In particular, the initial argument was phrased in terms of ``sufficiently small steps.'' At the author's request, the AI helped make this precise through the elementary-step decomposition and the continuity argument at step endpoints. The author subsequently checked and verified the complete argument, including all proofs, calculations, and cited claims.

Claude Opus 4.7 and ChatGPT Astra also assisted with generating the figures. The author takes full responsibility for every claim, proof, and citation in this manuscript.

\bibliographystyle{plainurl}
\bibliography{references}
\newpage
\appendix
\section{Proof of the Elementary-Step Decomposition}
\label{app:events}

\begin{proof}[Proof of Lemma~\ref{lem:step-decomposition}]
Fix a request point $r$ and the offline configuration $Y$ during the online movement phase.
If an online server already occupies $r$, the phase has zero duration and there are no elementary steps.
Otherwise, each server follows its initial path towards $r$ without reversing direction.
Movement towards $r$ may be upward or downward relative to the root.

Suppose that a server at $x$ is strictly obstructed by another server at $y\in(x,r)$.
Until service ends, the obstructing server stays in $[y,r)$, whether it moves or stops.
Any departure from $x$ would first pass through $[x,y)$ while still strictly obstructed,
which the movement rule forbids. Thus a strictly obstructed server remains stationary
for the rest of the phase.

If several servers occupy an unobstructed location, one representative
moves and the others remain stationary. As it leaves, the others become
strictly obstructed. The same convention applies when moving
servers meet: if their common location is unobstructed, one continues and the others stop;
otherwise, all stop. Thus a stationary server never resumes movement, and each initially moving
server stops at most once, so there are at most \(k\) stopping events.

Each moving server follows a simple path towards \(r\), so it reaches each
original vertex at most once and there are at most \(k|V|\) such events.
This includes a server reaching a vertex and thereby entering another
server's path to \(r\), obstructing it before they meet. A moving server also cannot newly meet a server that was already stationary immediately beforehand, since that server would already obstruct its path to $r$. Hence every new meeting before service ends involves at least two moving servers and causes at least one of them to stop permanently. Since each server stops at most once, there are only finitely many such meetings.

Between consecutive times at which a moving server reaches
a vertex or two servers meet, each moving server stays in
the interior of one edge. During such an interval, no server
passes another or moves onto another server's path to $r$.
Thus the obstruction relation and the set of moving servers
remain unchanged.

Let \(m(t):=\min_i\dist(x_i(t),r)\) be the distance from the request to the
closest online server at time \(t\). Before service ends, a closest server is
unobstructed, since any server lying strictly between it and \(r\) would be
even closer to \(r\). Thus at least one server at distance \(m(t)\) moves
toward \(r\) at unit speed, while no server can decrease its distance to \(r\)
faster than unit speed. Hence between consecutive events \(m(t)\) decreases
exactly at unit speed. Since server positions, and therefore \(m(t)\), are
continuous at event times, \(m(t)=m(0)-t\) until a server reaches \(r\).
Thus the movement phase ends after time \(L=m(0)<\infty\).

Partition $[0,L]$ at every time when a moving server reaches a vertex or when the moving set or obstruction relation changes. Further subdivide whenever a moving online server passes a location occupied by an offline server. Since each online server moves continuously towards $r$ without reversing direction, it can pass each fixed offline location at most once, so only finitely many additional subdivision times are introduced. Including the endpoints, write the resulting distinct times as $0=t_0<t_1<\cdots<t_N=L$. If several changes occur at the same time, that time is included only once.

Throughout the interior of each interval $[t_{q-1},t_q]$,
the set of moving servers and the obstruction relation
remain unchanged. Each moving server stays in the interior
of one original edge, moves either toward or away from
$\rho$ throughout the interval, and passes no offline server.
Thus $[t_{q-1},t_q]$ is an elementary step of duration
$\delta=t_q-t_{q-1}>0$.
\end{proof}

\section{Proof of the Tree Transportation Identity}
\label{app:transport-proof}

Throughout this appendix, points are counted with multiplicity. When several
points occupy the same location, we distinguish them by their labels.

\begin{proof}[Proof of Lemma~\ref{lem:transport}]
Write \(X=\{x_1,\ldots,x_k\}\) and \(Y=\{y_1,\ldots,y_k\}\), and recall that
\[
M(X,Y)=\min_{\sigma\in S_k}\sum_{i=1}^k \dist(x_i,y_{\sigma(i)}),
\qquad
D(z)=|X\cap T_z|-|Y\cap T_z|.
\]

\noindent\emph{Lower bound.}
Fix a matching \(\sigma\in S_k\). Since the length of a path is the integral
of its indicator function,
\[
\sum_{i=1}^k \dist(x_i,y_{\sigma(i)})
=
\int_T \sum_{i=1}^k
\mathbf{1}_{\{z\in[x_i,y_{\sigma(i)}]\}}\,dz.
\]
For a point \(z\notin X\cup Y\) in an edge interior, let \(a(z)\) be the
number of matched
pairs \((x_i,y_{\sigma(i)})\) with \(x_i\in T_z\) and
\(y_{\sigma(i)}\notin T_z\), and let \(b(z)\) be the number with
\(x_i\notin T_z\) and \(y_{\sigma(i)}\in T_z\). Pairs whose endpoints are
both in \(T_z\) contribute equally to \(|X\cap T_z|\) and \(|Y\cap T_z|\),
while pairs whose endpoints are both outside contribute to neither count.
Thus \(D(z)=a(z)-b(z)\). A matching path crosses \(z\) exactly when its
endpoints lie on opposite sides of the cut, so
\[
\sum_{i=1}^k \mathbf{1}_{\{z\in[x_i,y_{\sigma(i)}]\}}
=a(z)+b(z)\ge |a(z)-b(z)|=|D(z)|.
\]
Integrating and minimizing over \(\sigma\) gives
\(M(X,Y)\ge \int_T |D(z)|\,dz\).

\medskip
\noindent\emph{Upper bound.}

Before constructing the matching, subdivide every edge at each point of
\(X\cup Y\). This produces a finite rooted weighted tree with the same
metric. It leaves every matching cost and \(\int_T |D(z)|\,dz\) unchanged,
since subdivision only splits an edge integral into integrals over the
resulting subedges. Hence we may assume that every point of \(X\cup Y\) is a vertex, with edges and vertices henceforth referring to this subdivided tree. The finitely many vertices have measure zero, so the values of the integrand there do not affect the integral.

We now construct a matching attaining the lower bound. For a vertex \(v\),
let \(S_v\) be the subtree rooted at \(v\), including \(v\), and set
\(D_v:=|X\cap S_v|-|Y\cap S_v|\).

This matching is constructed recursively from the leaves toward the root.
At each leaf, match as many points of \(X\) there with points of \(Y\) there
as possible. Now let \(v\) be a vertex whose child subtrees have already
been treated. Among the points in those child subtrees that are still
unmatched, together with any points located at \(v\), match points of \(X\)
with points of \(Y\) for as long as both are available. Continue in this
way until the root is reached.

It is easy to see that, after the construction has been carried out inside
\(S_v\), exactly \(|D_v|\) points of \(S_v\) remain unmatched: they belong
to \(X\) if \(D_v>0\), and to \(Y\) if \(D_v<0\). Indeed, every match removes
one point from each configuration, so the difference between the numbers of
unmatched \(X\) and \(Y\) points in \(S_v\) remains
\(|X\cap S_v|-|Y\cap S_v|=D_v\), and we continue matching until one side is
exhausted. At the root, \(D_\rho=|X|-|Y|=0\), so every point is matched and
the construction gives a perfect matching \(\sigma^*\).

Now let \(e\) be the edge joining a non-root vertex \(v\) to its parent.
After the construction has been carried out inside \(S_v\), exactly
\(|D_v|\) points of \(S_v\) remain unmatched, all from the same configuration.
Each of these points must therefore eventually be matched to a point outside
\(S_v\). Conversely, every point that has already been matched inside \(S_v\)
is paired with another point of \(S_v\). Hence exactly \(|D_v|\) matched pairs
have one endpoint in \(S_v\) and the other outside \(S_v\). The unique path
joining the two endpoints of such a pair crosses \(e\), while a path joining
two points of \(S_v\) does not. Therefore exactly \(|D_v|\) matching paths
cross \(e\).

Let \(w_e\) be the length of \(e\). For every \(z\) in the interior of
\(e\), \(T_z\) and \(S_v\) contain the same points of each configuration,
since the edge interior contains no point of \(X\cup Y\). Thus \(D(z)=D_v\)
throughout the edge interior, so \(e\) contributes \(w_e|D_v|\) to
\(\int_T |D(z)|\,dz\). Exactly \(|D_v|\) matching paths cross \(e\), giving
the same contribution to the matching cost. Summing over all edges gives
\[
\sum_{i=1}^k \dist(x_i,y_{\sigma^*(i)})
=
\int_T |D(z)|\,dz.
\]
Thus \(M(X,Y)\le \int_T |D(z)|\,dz\). Together with the lower bound,
this proves the identity.
\end{proof}

\section{Continuity at Step Endpoints}
\label{app:continuity}

Throughout this appendix, the offline configuration $Y$ is fixed while
\DC{} moves toward the current request. At the boundaries between
elementary steps, the set of moving servers, the obstruction relation,
or the height ordering of the taxis may change. We show that none of
these changes causes a jump in the potentials used above. Thus the estimates proved in the interiors of elementary
steps remain valid when the endpoints are included.

\begin{lemma}
The functions $M(X,Y)$, $P(X)$, and $\Theta(X)$ are continuous as the
online servers move continuously. The same is true of $\Psi(X)$ in the
HST setting and of $\Phi_{\mathrm{tree}}(X,Y)$ for general weighted
trees. Hence all the potentials used in Sections~4--6 are continuous
throughout the online movement.
\end{lemma}

\begin{proof}
\emph{The basic terms $M$, $P$, and $\Theta$.} Give the servers fixed labels and write their positions at time
$t$ as $x_1(t),\ldots,x_k(t)$. For two times $s$ and $t$, write
$R(s,t):=\sum_{i=1}^k \dist(x_i(s),x_i(t))$ for the total distance needed
to move each labelled server from its position at time $s$ to its position
at time $t$.

The triangle inequality gives
$|M(X(t),Y)-M(X(s),Y)|\le R(s,t)$,
$|P(X(t))-P(X(s))|\le (k-1)R(s,t)$, and
$|\Theta(X(t))-\Theta(X(s))|\le R(s,t)$.

For the first inequality, fix an optimal matching at time $s$ and use the same matching at time $t$; its cost changes by at most $R(s,t)$. Reversing the roles of $s$ and $t$ gives the absolute value bound. The other two inequalities follow directly from the triangle inequality. Thus $M$, $P$, and $\Theta$ are Lipschitz in the server positions, and hence continuous.

\emph{The HST potential $\Psi$.}
In the HST setting, consider two times $s$ and $t$, and let $\varepsilon$ be the maximum distance that any online taxi moves between them, that is,
\[
\varepsilon:=\max_{1\le j\le k}\dist(x_j(s),x_j(t)).
\]
Since $h(x)=W-\dist(\rho,x)$, the height of each taxi changes by at most $\varepsilon$.

Let $h_0(t)\le\cdots\le h_{k-1}(t)$ denote the taxi heights in increasing order. Then $|h_i(t)-h_i(s)|\le\varepsilon$ for every $0\le i<k$. Indeed, at time $s$ there are at least $i+1$ taxis of height at most $h_i(s)$. At time $t$, those same taxis have height at most $h_i(s)+\varepsilon$, so $h_i(t)\le h_i(s)+\varepsilon$. Reversing the roles of $s$ and $t$ gives the opposite inequality.

Recall that $\chi_\ell(h)=h-\alpha_\ell$ for
$\alpha_\ell\le h\le\alpha_{\ell+1}$ and is constant for
$h\le\alpha_\ell$ and for $h\ge\alpha_{\ell+1}$.
Thus $|\chi_\ell(v)-\chi_\ell(w)|\le |v-w|$.
It follows that each term $c_{i,\ell}\chi_\ell(h_i)$ changes by at most $c_{i,\ell}\varepsilon$, and therefore
\[
|\Psi(X(t))-\Psi(X(s))|
\le
\varepsilon\sum_{i=0}^{k-1}\sum_{\ell=0}^{d-1}c_{i,\ell}.
\]
This estimate also applies when two taxis exchange ranks,
since their heights agree at that time and sorting the
heights introduces no discontinuity. Likewise, crossing a layer boundary causes no jump because each $\chi_\ell$ is continuous at its endpoints. Thus $\Psi$ is Lipschitz in the taxi positions, and hence continuous.

\emph{The general-tree potential $\Phi_{\mathrm{tree}}$.}
Recall from~\eqref{eq:tree-motion-bound} that if one configuration moves a total distance $C$ while the other is fixed, then
$|\Delta\Phi_{\mathrm{tree}}|\le M_d C$.
Consider the paths joining the labelled online taxis' positions at times $s$ and $t$. The total distance moved is $R(s,t)$, so
\[
|\Phi_{\mathrm{tree}}(X(t),Y)-\Phi_{\mathrm{tree}}(X(s),Y)|
\le M_dR(s,t).
\]
Thus $\Phi_{\mathrm{tree}}$ is Lipschitz in the taxi positions, and hence continuous. The bound also applies when a taxi crosses between edges of different depths, since~\eqref{eq:tree-motion-bound} applies along the entire path.

The full potentials in Sections~\ref{sec:server} and~\ref{sec:hst} are fixed linear combinations of the continuous quantities above, and are therefore continuous as well.
\end{proof}

We now use continuity to pass from the interior of an elementary step to the whole step. Let $[a,b]$ be an elementary step, rank-refined in the HST setting, and suppose that $F(t)-F(s)\le \gamma(t-s)$ for every $a<s<t<b$, where $F$ is one of the continuous quantities considered above, or a fixed linear combination of them. Letting $s\downarrow a$ and $t\uparrow b$, continuity of $F$ gives $F(b)-F(a)\le \gamma(b-a)$. The same argument applies when the relation is an equality.

For the online movement inequalities, let $m=|U|+|B|$ be the number of movers in the step interior. Between times $s$ and $t$, each mover travels distance $t-s$, so the movement cost is $m(t-s)$. Hence an interior estimate $m(t-s)+F(t)-F(s)\le0$ extends, by the same limiting argument, to $m(b-a)+F(b)-F(a)\le0$. Changes occurring exactly at $a$ or $b$ take no time and therefore contribute no additional movement cost.

\end{document}